\documentclass[a4paper,11pt]{article}
\usepackage{hyperref}
\usepackage[top=1 in,bottom=1 in,left=1.0 in,right=1.0 in]{geometry}
\usepackage{float}
\usepackage{cite}
\usepackage{epic,graphicx}
\usepackage{subfigure}
\usepackage{amssymb,amsmath,amsthm,amscd}
\usepackage{mathrsfs}
\usepackage{tikz}

\usepackage{color}
\newcommand{\bred}{\begin{color}{red}}
\newcommand{\ecl}{\end{color}}
\newcommand{\bblue}{\begin{color}{blue}}
\newcommand{\bgre}{\begin{color}{green}}
\newcommand{\bora}{\begin{color}{orange}}

\def\m@th{\mathsurround=0pt}
\mathchardef\bracell="0365
\def\upbrall{$\m@th\bracell$}
\def\undertilde#1{\mathop{\vtop{\ialign{##\crcr
    $\hfil\displaystyle{#1}\hfil$\crcr
     \noalign
     {\kern1.5pt\nointerlineskip}
     \upbrall\crcr\noalign{\kern1pt
   }}}}\limits}

\def \h#1{\widehat{#1}}
\def \t#1{\widetilde{#1}}
\def \b#1{\overline{#1}}
\def \d#1{\dot{#1}}
\def \tb#1{\widetilde{\overline{#1}}}
\def \th#1{\widehat{\widetilde{#1}}}
\def \hb#1{{\widehat{\overline{#1}}}}

\def \thb#1{\widehat{\widetilde{\overline{#1}}}}

\newtheorem{lemma}{Lemma}

\newtheorem{theorem}{Theorem}
\newtheorem{remark}{Remark}
\newcommand{\bs}{\boldsymbol}

\newcommand{\bsb}{\begin{subequations}}
\newcommand{\esb}{\end{subequations}}

\newcommand{\bK}{\boldsymbol{K}}

\newcommand{\bM}{\boldsymbol{M}}
\newcommand{\br}{\boldsymbol{r}}
\newcommand{\bc}{\boldsymbol{c}}
\newcommand{\bI}{\boldsymbol{I}}
\newcommand{\bT}{\boldsymbol{T}}

\usepackage{mathrsfs}
\usepackage{graphicx}
\usepackage{bm}

\begin{document}

\title{Symmetric discrete AKP and  BKP equations }
\author{Shangshuai Li$^1$,~Frank W. Nijhoff$^2$,
~Ying-ying Sun$^3$, ~Da-jun Zhang$^1$\footnote{Corresponding author. Email: djzhang@staff.shu.edu.cn}\\
{\small $~^1$Department of Mathematics, Shanghai University, Shanghai 200444,  China}\\
{\small $~^2$School of Mathematics, University of Leeds, Leeds LS2 9JT, United Kingdom}\\
{\small $~^3$Department of Mathematics, University of Shanghai for Science and Technology, Shanghai 200093,  China}}

\maketitle

\begin{abstract}
We show that when  KP (Kadomtsev-Petviashvili) $\tau$ functions allow special symmetries,
the discrete BKP equation can be expressed as a linear combination of the discrete AKP equation
and its reflected symmetric forms.
Thus the discrete AKP and BKP equations can share  the same $\tau$ functions with these symmetries.
Such a connection is extended to 4 dimensional   (i.e. higher order)  discrete AKP  and BKP equations in the corresponding
discrete hierarchies.
Various explicit forms of such  $\tau$ functions, including Hirota's form, Gramian, Casoratian
and polynomial, are given.
Symmetric $\tau$ functions of Cauchy matrix form that are composed of Weierstrass $\sigma$ functions
are  investigated. As a result we obtain a  discrete BKP equation with elliptic coefficients.

\begin{description}
\item[PACS numbers:] 02.30.Ik, 02.30.Ks, 05.45.Yv
\item[Keywords:] discrete AKP, discrete BKP, symmetric  $\tau$ function, solution, elliptic function
\end{description}
\end{abstract}

\section{Introduction}\label{sec-1}

The Kadomtsev-Petviashvili (KP) equation is one of the most famous (2+1)-dimensional integrable systems.
In discrete case, the discrete AKP (dAKP) equation and BKP (dBKP) equation are two master equations in the KP family.
With parameterised coefficients, they are, respectively,
\begin{equation}\label{AKP}
    A \doteq    (a-b)\th\tau \b\tau+(b-c)\hb \tau \t \tau+(c-a)\tb \tau \h \tau=0,
\end{equation}
and
\begin{align}
   B \doteq   &  (a-b)(b-c)(c-a)\th{\b \tau} \tau+(a-b)(a+c)(b+c)\th \tau \b \tau  \notag  \\
        &~~~ +(b-c)(b+a)(c+a)\hb\tau \t \tau+(c-a)(c+b)(a+b)\tb\tau \h \tau=0.   \label{BKP}
\end{align}
Here $a,b,c$ are spacing parameters and tilde, hat, bar serve as notations of shifts in different directions (see \eqref{shift}).
An alternative form of Eq.(1) is
\begin{equation}\label{AKP-alt}
        c(a-b)\th\tau \b\tau+a(b-c)\hb \tau \t \tau+b(c-a)\tb \tau \h \tau=0,
\end{equation}
which is connected with \eqref{AKP} by changing
\begin{equation}\label{abc}
       (a, b, c) \longrightarrow (1/a, 1/b, 1/c).
\end{equation}
Note that Eq.(2) does not change under the above replacement.
The coefficients in both the dAKP and dBKP equation can be arbitrary nonzero numbers if we do not require
$\tau=1$ is a solution.
In fact, all the coefficients $z_i$ of the following dAKP equation
\begin{equation}\label{AKP-z}
       z_1 \hb \tau \t \tau+ z_2\tb \tau \h \tau + z_3 \th\tau \b\tau =0,
\end{equation}
can be gauged to be $1$ by $\tau\to z_1^{-ml}z_2^{-nl}z_3^{-nm}\tau'$ (cf.\cite{Nimmo-JPA-1997,Nimmo-CSF-2000}),
and for the dBKP equation
\begin{equation}\label{BKP-z}
       z_1 \hb \tau \t \tau+ z_2\tb \tau \h \tau + z_3 \th\tau \b\tau + z_4 \th{\b \tau} \tau =0,
\end{equation}
the transformation is (cf.\cite{NimS-PRLSA-1997})
\begin{equation}\label{gauge-bkp}
\tau\to \biggl(\frac{z_2z_3}{z_1z_4}\biggr)^{\frac{ml}{2}}\biggl(\frac{z_1z_3}{z_2z_4}\biggr)^{\frac{nl}{2}}
\biggl(\frac{z_1z_2}{z_3z_4}\biggr)^{\frac{nm}{2}}\tau'.
\end{equation}

Eq.\eqref{AKP} originated from Hirota's discrete analogue of the generalized Toda equation (DAGTE) \cite{Hirota-JPSJ-1981}
that was parameterised later by Miwa \cite{Miwa-PJA-1982} for the sake of expression of $N$-soliton solutions.
It is also known as the discrete KP equation, the Hirota equation, or the Hirota-Miwa equation.
Note that the DAGTE and the dAKP equation are equivalent in the sense that there exist a set of parameter transformations
to transform them to each other \cite{Hirota-RIMS-2014}.
Equation \eqref{BKP} was first given by Miwa \cite{Miwa-PJA-1982} and now bears his name.
It also appears as a nonlinear superposition formula of the (2+1) dimensional sine-Gordon system \cite{NimS-PRLSA-1997}.
Both equations have Lax triads \cite{Nimmo-JPA-1997,Nimmo-CSF-2000,NimS-PRLSA-1997},
and both equations are 4D consistent \cite{AdlBS-IMRN-2012,AdlBS-CMP-2003}.

Both the dAKP and dBKP equations have $N$-soliton solutions, which are possible to be written out from
those of the continuous AKP and BKP equation (cf.\cite{DatKM-PJA-1981})
by means of Miwa's transformation \cite{Miwa-PJA-1982}.
Let us present these solutions by the following uniform formula,
\begin{equation}\label{tau}
\tau=\sum_{\mu = 0,1}\mathrm{exp}
\left[\sum^N_{j=1}\mu_j \eta_j+ \sum^{N}_{1\leq i<j}\mu_i\mu_j a_{ij} \right],
\end{equation}
where $\mathrm{e}^{a_{ij}}=A_{ij}$,  and the summation over $\mu$ means to take all
possible $\mu_j= {0,1}$ $(j = 1,2,?????? ,N)$.
For the dAKP equation \eqref{AKP},
\begin{equation}
\mathrm{e}^{\eta_i}=\biggl(\frac{a-q_i}{a-p_i}\biggr)^{n}\biggl(\frac{b-q_i}{b-p_i}\biggr)^{m}
\biggl(\frac{c-q_i}{c-p_i}\biggr)^{l}\eta_i^{(0)},
~~A_{ij}=\frac{(p_i-p_j)(q_i-q_j)}{(p_i-q_j)(q_i-p_j)},
\label{PWF-AKP}
\end{equation}
while for the dBKP equation \eqref{BKP},
\begin{subequations}\label{PWF-BKP}
\begin{align}
&\mathrm{e}^{\eta_i}=\biggl[\frac{(a-p_i)(a-q_i)}{(a+p_i)(a+q_i)}\biggr]^{n}
\biggl[\frac{(b-p_i)(b-q_i)}{(b+p_i)(b+q_i)}\biggr]^{m}
\biggl[\frac{(c-p_i)(c-q_i)}{(c+p_i)(c+q_i)}\biggr]^{l}\eta_i^{(0)},\\
&A_{ij}=\frac{(p_i-p_j)(p_i-q_j)(q_i-p_j)(q_i-q_j)}{(p_i+p_j)(p_i+q_j)(q_i+p_j)(q_i+q_j)},
\end{align}
\end{subequations}
where $\eta_i^{(0)}\in \mathbb{C}$.
For the continuous AKP and BKP hierarchies, their $\tau$ functions can often  be connected as \cite{DatJKM-PD-1982}
\begin{equation}\label{Pfaffian}
\tau^2_{\mathrm{BKP}}=\tau_{\mathrm{AKP}}^{}|_{\mathrm{odd}}
\end{equation}
where the right hand side means the AKP $\tau$ function  with all the even-index coordinates removed.
  The relation \eqref{Pfaffian}  is thought to be true for the dAKP and dBKP under Miwa's coordinates, but has
strictly been shown only for
special (mostly soliton type) solutions.
The perfect square structure implies solutions in terms of Pfaffians for the (discrete) BKP hierarchy
\cite{FuN-PRSA-2017,Hirota-JPSJ-1989,TsujH-JPSJ-1996}.

In this paper, we would like to investigate a different connection between the dAKP equation and dBKP equation.
We will show that when $\tau_{\mathrm{AKP}}^{}$ has some special symmetries (see \eqref{sym})
these two equations can share the same $\tau$ function.
We will see that with these symmetries
the dBKP equation can be expressed as a linear combination of the dAKP equation
and its symmetric deformations.
Such a connection can also be extended to 4 dimensional (4D) dAKP  and dBKP equations, i.e. the higher order
equations in the corresponding discrete hierarchies\footnote{These equations are not really four-dimensional, as
their integrabiity only holds modulo the lower equations in the hierarchy.}.
Various explicit forms of such  $\tau$ functions will be given. In a sense these solutions of the dBKP
equation could be considered to be `reducible' as they are also simultaneously solutions of
the dAKP equation\footnote{This is reminiscent of the reducible, i.e. special, solutions of the Painlev\'e equations which hold
for special parameter values of those equations, for which also simultanesouly a lower order equation is valid
and which can be expressed in terms of classical special functions.}.
We will also investigate symmetric $\tau$ functions in Cauchy matrix form that are composed of Weierstrass $\sigma$ functions.
As  its most significant consequence, a dBKP equation with elliptic coefficients will be obtained,
which is conjectured to be the canonical form for the elliptic solutions (including those which are
not shared with the dAKP equation) of the dBKP equation.

The paper is organized as follows.
In Sec.\ref{sec-2} we discuss symmetric $\tau$ functions  and connections between dAKP, dBKP,
4D dAKP and 4D dBKP equations.
In Sec.\ref{sec-3} various explicit forms of symmetric $\tau$ functions are given,
including Hirota's form, Gramian, Casoratian and polynomial.
Then in Sec.\ref{sec-4} we investigate symmetric $\tau$ functions composed of Weierstrass $\sigma$ function
and present  a dBKP equation with elliptic coefficients.
Finally, conclusions are given in Sec.\ref{sec-5}.

\section{Discrete AKP  and  BKP}\label{sec-2}

Let $f(n_1,n_2,\cdots)$ be a function $f: \mathbb{Z}^{+\infty} \mapsto \mathbb{C}$,
and $(a_1, a_2, \cdots)$ corresponding spacing parameters associated with each lattice variable $(n_1,n_2, \cdots)$.
Define shift operator $E_{n_i}$ by
\[E_{n_i} f(n_1,n_2,\cdots)= f(n_1,\cdots, n_{i-1}, n_i+1, n_{i+1}, \cdots).\]
For convenience, we denote $(n_1,n_2,n_3,n_4)=(n,m,l,h)$, $(a_1, a_2, a_3, a_4)=(a,b,c,d)$,
and the shifts by
\begin{equation}
\t f=f(n+1,m,l,h),~ \h f=f(n,m+1,l,h), ~\b f=f(n,m,l+1,h),~ \d f=f(n,m,l,h+1).
\label{shift}
\end{equation}
We also employ the following notations with spacing parameters when necessary,
\[\tau(n,m,l)\doteq \tau((n,a),m,l) \doteq \tau(n,(m,b),l) \doteq \tau(n,m,(l,c)).\]

The celebrated Sato theory \cite{Sato-1981,MiwJD-2000} introduces plane wave factors
 $\mathrm{e}^{\xi(\mathbf{t},k)}$ with a universal dispersion relation
(here $\mathbf{t}=(t_1,t_2,\cdots)$)
\begin{equation}
\xi(\mathbf{t},k)=\sum_{i=1}^{\infty}k^it_i.
\label{xi}
\end{equation}
Miwa introduced, for instance, for the dAKP equation, the discrete plane wave factors \cite{Miwa-PJA-1982}
\begin{equation}\label{pwf}
\prod_{i=1}^{\infty}\Bigl(\frac{1-a_ip}{1-a_iq}\Bigr)^{n_i}
\end{equation}
which can be  rewritten in a continuous exponential function form by  $\mathrm{e}^{-\xi(\mathbf{x},p)+\xi(\mathbf{x},q)}$,
where
\begin{equation}
x_j=\frac{1}{j}\sum^{\infty}_{i=1}n_i a_i^j
\end{equation}
and $\mathbf{x}=(x_1, x_2, \cdots)$ are usually referred to as Miwa's coordinates.
Discrete integrable systems can be investigated by means of Sato's theory with Miwa's coordinates.

\subsection{Reflection symmetric $\tau$ function, dAKP  and  dBKP}\label{sec-2-1}

Suppose that $\tau(n,m,l)$ has the following symmetries
\begin{subequations}\label{sym}
\begin{align}
& \tau((n,a),m,l)=\tau((-n,-a),m,l),\label{sym-a}\\
& \tau(n,(m,b),l)=\tau(n,(-m,-b),l),\label{sym-b}\\
& \tau(n,m,(l,c))=\tau(n,m,(-l,-c)).\label{sym-c}
\end{align}
\end{subequations}
If $\tau$ is a solution of the dAKP equation \eqref{AKP}, then  the above symmetries lead to the
reflected dAKP equations, respectively,
\begin{subequations}\label{ssym}
\begin{align}
A_1& \doteq (b-c)\th{\b \tau}\tau+(c+a)\th \tau \b\tau-(a+b)\tb \tau \h\tau=0, \label{ssym-a}\\
A_2& \doteq (c-a)\th{\b \tau}\tau+(a+b) \hb \tau \t \tau-(b+c)\th \tau\b \tau=0, \label{ssym-b}   \\
A_3& \doteq (a-b)\th{\b \tau}\tau+(b+c)\tb \tau \h \tau-(c+a)\hb \tau \t \tau=0. \label{ssym-c}
\end{align}
\end{subequations}

In fact, if $\tau$ has  symmetry \eqref{sym-a}, it then follows from \eqref{AKP} that
\[  (-a-b)\undertilde{\h\tau} \b\tau+(b-c)\hb \tau \undertilde\tau+(c+a)\undertilde{\b \tau} \h \tau=0.\]
By a up shift in tilde direction we immediately get \eqref{ssym-a}. \eqref{ssym-b}
and \eqref{ssym-c} can be derived in a similar way.

Then we present the following connection between the dAKP and dBKP equations.

\begin{theorem}\label{T:1}
 If  $\tau$ satisfies the dAKP equation \eqref{AKP} and has symmetries \eqref{sym},
 then  $\tau$ satisfies the dBKP equation \eqref{BKP}.
\end{theorem}

\begin{proof}
For the dAKP $\tau$ function, when it  has symmetries  \eqref{sym},
$\tau$ satisfies \eqref{AKP} and \eqref{ssym} simultaneously. Direct calculation gives rise to
\begin{align}
&(c-a)(a-b)\times A_1 + (a-b)(b-c)\times A_2 + (b-c)(c-a)\times A_3 \nonumber\\
& + (a^2+b^2+c^2+3ab+3bc+3ca)\times A  \nonumber\\
=~& 3 \times B.\label{BKPA}
\end{align}
\end{proof}

\begin{remark}\label{R:1}
Among the dAKP equation \eqref{AKP} and the reflected dAKP equations \eqref{ssym},
$A, A_1, A_2, A_3$ are not linearly independent.
Apart from the relation
\[A_1+A_2+A_3=A,\]
any element in  $\{A, A_1, A_2, A_3\}$
can be a linear combination of any two elements of the same set.
For example,
\[A=\frac{a-c}{a+b}A_1+\frac{b-c}{a+b}A_2,~~
A_3=-\frac{b+c}{a+b}A_1-\frac{a+c}{a+b}A_2.\]
This indicates that there are alternative  expressions of \eqref{BKPA}
in terms of only two elements of  $\{A, A_1, A_2, A_3\}$.
For example, we have
\begin{equation}
B=(a+b) (a+c) A-(a-b) (a-c) A_1=(a-c) (b+c) A_1-(b-c) (a+c) A_2,
\label{BKPA-alt}
\end{equation}
which is simpler than \eqref{BKPA}.
\end{remark}

\begin{remark}\label{R:2}
  Note that the symmetries can be extended to
\begin{subequations}\label{sym'}
\begin{align}
& \tau((n,a),m,l)=\gamma_1 A_1^nB_1^mC_1^l \,\tau((-n,-a),m,l),\label{sym'-a}\\
& \tau(n,(m,b),l)=\gamma_2 A_2^nB_2^mC_2^l \,\tau(n,(-m,-b),l),\label{sym'-b}\\
& \tau(n,m,(l,c))=\gamma_3 A_3^nB_3^mC_3^l \,\tau(n,m,(-l,-c)),\label{sym'-c}
\end{align}
\end{subequations}
where $A_i, B_i, C_i, \gamma_i$ are nonzero constants.
Due to the gauge property of discrete bilinear equations (e.g.\cite{HieZ-JPA-2009}),
Theorem \ref{T:1} and \eqref{ssym} are still valid if replacing \eqref{sym} with \eqref{sym'}.
\end{remark}

\subsection{4D  dAKP}\label{sec-2-2}

The 4D dAKP equation is given in \cite{OhtHTI-JPSJ-1993} via a compact form
\begin{equation}\label{4DAKP}
\left |
\begin{array}{cccc}
a^2 & a & 1 & \t \tau \d{\hb \tau} \\
b^2 & b & 1 & \h \tau \d{\tb \tau} \\
c^2 & c & 1 & \b \tau \d{\th \tau} \\
d^2 & d & 1 & \d \tau \thb \tau
\end{array}
\right |=0,
\end{equation}
which has an explicit expression
\begin{align}
   & (b-c)(c-d)(b-d)\t\tau  \d{\hb \tau} -(a-c)(c-d)(a-d)\h \tau \d{\tb \tau} \notag \\
   & ~~~~ +(a-b)(b-d)(a-d)\b \tau \d{\th \tau} -(a-b)(b-c)(a-c)\d \tau \thb \tau =0.\label{4dAKP}
\end{align}
The dAKP equation \eqref{AKP} is 4D consistent \cite{AdlBS-IMRN-2012}.
We note that the above 4D dAKP equation is a result of 4D consistency of \eqref{AKP}.
In fact, embedding eight copies of the dAKP equations on a hypercube (see Fig.\ref{F:1})
and then calculating a triplly shifted $\tau$ with proper initial points,
(e.g. calculating $\d{\th \tau}$ with initials $\t \tau, \h\tau, \b\tau, \d\tau, \tb\tau, \d{\t\tau}, \d{\tb \tau},  \d{\hb \tau}$),
one gets the 4D dAKP equation \eqref{4dAKP}.
In addition to \eqref{4dAKP},  the dAKP equation defined on three elementary cubes can yield another
4D lattice equation (cf. Eq.(29) in \cite{AdlBS-IMRN-2012})
\begin{equation}\label{4dAKP-A0}
\mathcal{A}_0\doteq (a-b)(c-d)\th\tau \d{\b\tau}-(a-c)(b-d)\d{\h\tau}\tb\tau+(a-d)(b-c)\d{\t\tau}\hb\tau.
\end{equation}

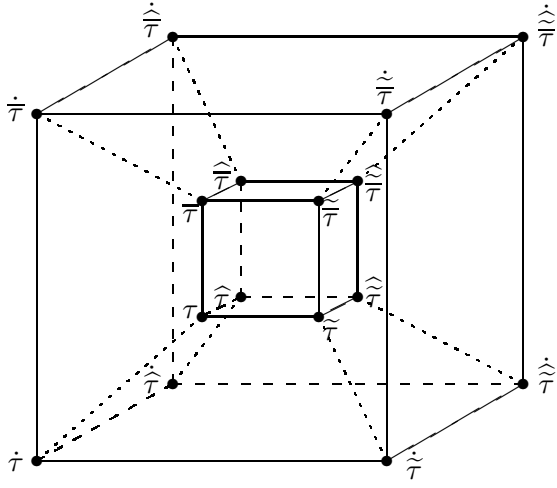
\begin{figure}
\setlength{\unitlength}{0.0010in}
\begin{picture}(2400,2400)(-1500,-50)
  \put(1800,  0){\circle*{60}} \put(0  ,1800){\circle*{60}}
  \put( 700, 400){\circle*{60}} \put(700,2200){\circle*{60}}
  \put(  0,  0){\circle*{60}}  \put(1800,1800){\circle*{60}}
  \put( 2500,400){\circle*{60}}  \put(2500, 2200){\circle*{60}}
  \put( 0,  0){\line(1,0){1800}}
  \put( 0,1800){\line(1,0){1800}}
  \put(700,2200){\line(1,0){1800}}
  \put(  0, 0){\line(0,1){1800}}
  \put(1800, 0){\line(0,1){1800}}
  \put(2500,400){\line(0,1){1800}}

\drawline(  0,1800)(700,2200)
\drawline(1800,1800)(2500, 2200)
\drawline(1800,  0)(2500, 400)
\dashline{50}(0,  0)(700, 400)
 \dashline{50}(700,400)(0,0)
  \dashline{50}(700,400)(2500,400)
\dashline{50}(700,400)(700,2200)

\drawline(850,  750)(1450 ,750)
\drawline(1650, 850)(1450 ,750)
\drawline(850,  750)(850,1350)
\drawline(1450,  1350)(850,1350)
\drawline(1450,  1350)(1450,750)
\drawline(1050,  1450)(850,1350)
\drawline(1050,  1450)(1650,1450)
\drawline(1450,  1350)(1650,1450)
\drawline(1650,  850)(1650,1450)
\dashline{50}(850,  750)(1050 ,850)
\dashline{50}(1650,  850)(1050 ,850)
\dashline{50}(1050,  1450)(1050 ,850)

\put(850,  750){\circle*{60}} \put(1450 ,750){\circle*{60}}
  \put(1650, 850){\circle*{60}} \put(850,1350){\circle*{60}}
  \put(1450,  1350){\circle*{60}}  \put(1650,1450){\circle*{60}}
  \put(1050,  1450){\circle*{60}}  \put(1050 ,850){\circle*{60}}
\dashline{18}(1800,  0)(1450,  750)
\dashline{18}(0,  0)(850,  750)
\dashline{18}(700,  400)(1050, 850)
\dashline{18}(2500,400)(1650, 850)
\dashline{18}(2500,2200)(1650, 1450)
\dashline{18}(700,  2200)(1050, 1450)
\dashline{18}(1800,  1800)(1450,  1350)
\dashline{18}(0,  1800)(850,  1350)
      \put(-150,-50){$\dot \tau$}
         \put(1900, -95 ){$\dot{\t \tau}$}
      \put(1470,  630){$\t \tau$}
     \put(750,  740){$\tau$}

   \put(550,  350){$\dot{\h{\tau}}$}
     \put(910, 800){$\h \tau$}
      \put(2580,350){$\dot{\th \tau}$}
     \put(1690, 800){$\th \tau$}

     \put(2580,2150){$\dot{\thb{\tau}}$}
     \put(1690, 1350){$\thb \tau$}
      \put(550,  2200){$\dot{\hb \tau}$}
     \put(910, 1410){$\hb \tau$}

     \put(1750,  1850){$\dot{\tb{\tau}}$}
     \put(1470,  1220){$\tb \tau$}
      \put(-150,  1760){$\dot{\b \tau}$}
     \put(750,  1240){$\b \tau$}
\end{picture}
\vskip 0.3cm
\caption{4D  hypercube}\label{F:1}
\end{figure}

With regard to solutions, by virtue of 4D consistency,
if we  consistently extend the dAKP plane wave factor in \eqref{PWF-AKP}  to the 4th dimension,
then the resulted $\tau$ function will be a solution to the 4D equations \eqref{4dAKP} and \eqref{4dAKP-A0} as well.

\subsection{4D  dBKP}\label{sec-2-3}

The 4D dBKP equation is (see Eq.(2.4) in \cite{TsujH-JPSJ-1996})
\begin{align}
& (a-b)(a-c)(a-d)(b-c)(b-d)(c-d)\d{\thb \tau}\tau
-(a-b)(a+c)(a+d)(b+c)(b+d)(c-d)\th \tau \d{\b\tau}\nonumber\\
&+ (a+b)(a-c)(a+d)(b+c)(b-d)(c+d)\tb \tau \d{\h\tau}
- (a+b)(a+c)(a-d)(b-c)(b+d)(c+d)\hb \tau \d{\t\tau}\nonumber\\
= & \,0.
\label{BKP-4D}
\end{align}

There are at least three ways by which the above equation is connected with known lattice equations.
First, this equation is  a result of the 4D consistency of the dBKP equation \eqref{BKP} \cite{AdlBS-CMP-2003}.
Second, a more precise relation between the 4D dBKP and 3D dBKP can be given.
Denoting the  3D dBKP equation \eqref{BKP} by dBKP$(n,m,l)=0$ and
the 4D dBKP equation \eqref{BKP-4D} by 4DdBKP$(n,m,l,h)=0$, then we have (cf.\cite{Vek-JPA-2019})
\begin{align*}
& -\d \tau \times \mathrm{4DdBKP}(n,m,l,h)\\
~& (a+b)(a+c)(a-d)\d{\t \tau}\times \mathrm{dBKP}(h,m,l)
+(a+b)(b+c)(b-d)\d{\h \tau}\times \mathrm{dBKP}(n,h,l) \\
& + (a+c)(b+c)(c-d)\d{\b \tau}\times \mathrm{dBKP}(n,m,h)
+(a-d)(b-d)(c-d) \tau \times E_h \mathrm{dBKP}(n,m,l).
\end{align*}
The third way is the connection with the 4D dAKP equation \eqref{4dAKP} when $\tau$
allows  symmetries \eqref{sym} and
\begin{equation}
\tau(n,m,l,(h,d))=\tau(n,m,l,(-h,-d)).
\label{sym-d}
\end{equation}
With these symmetries the 4D dAKP equation \eqref{4dAKP} yields
\begin{align*}
\mathcal{A}_1 &\doteq (b-c)(c-d)(b-d) \tau  \d{\thb \tau} -(a+c)(c-d)(a+d)\th \tau \d{\b \tau}  \\
   & ~~~~ +(a+b)(b-d)(a+d)\tb \tau \d{\h \tau}-(a+b)(b-c)(c+a)\d{\t \tau} \hb \tau =0,\\
\mathcal{A}_2 &\doteq (a-c)(c-d)(a-d)  \tau \d{\thb \tau}-(b+c)(c-d)(d+b)\th\tau  \d{\b \tau}   \\
   & ~~~~ +(a+b)(b+d)(a-d)\hb \tau \d{\t \tau}-(a+b)(b+c)(a-c)\d {\h\tau} \tb \tau =0,\\
\mathcal{A}_3 &\doteq (a-b)(b-d)(a-d) \tau \d{\thb \tau} +(a+c)(c+d)(a-d)\hb \tau \d{\t \tau}  \\
   & ~~~~ - (b+c)(c+d)(b-d)\tb\tau  \d{\h \tau}-(a-b)(b+c)(c+a)\d{\b \tau} \th \tau =0,\\
\mathcal{A}_4 &\doteq (a-b)(b-c)(a-c) \tau \d{\thb \tau} +(a-c)(c+d)(d+a)\d{\h \tau} \tb \tau  \\
   & ~~~~ -(a-b)(b+d)(d+a) \d{\b \tau}\th \tau - (b-c)(c+d)(d+b) \d{\t\tau}\hb \tau=0.
\end{align*}
Together with \eqref{4dAKP-A0}, on can find
\begin{align*}
& 4 \times \mathrm{4DdBKP}(n,m,l,h)\\
=~& (a-b)(a-c)(a-d)\mathcal{A}_1 +(a-b)(b-c)(b-d)\mathcal{A}_2  + (a-c)(b-c)(c-d)\mathcal{A}_3\\
& +(a-d)(b-d)(c-d) \mathcal{A}_4 +[(a^2-b^2+c^2-d^2)^2-4(a+b)(b+c)(a+d)(c+d)]\mathcal{A}_0.
\end{align*}

Then, for solutions of the 4D dBKP equation \eqref{BKP-4D}, we have the following.
\begin{theorem}\label{T:2}
Once we have a dAKP $\tau$ function, we consistently extend its plane wave factor  to the 4th dimension
and impose  symmetries \eqref{sym} and \eqref{sym-d}.
Then the $\tau$ function is a solution to the 4D dBKP equation \eqref{BKP-4D}.
\end{theorem}

\section{$\tau$ functions with   symmetries}\label{sec-3}

We construct $\tau$ function that possesses symmetries \eqref{sym} or \eqref{sym'}.
Obviously, the  $\tau$ function in Hirota's form \eqref{tau} with \eqref{PWF-AKP}$|_{q_i=-p_i}$,
i.e.
\begin{equation}
\eta_i=\biggl(\frac{a+p_i}{a-p_i}\biggr)^{n}\biggl(\frac{b+p_i}{b-p_i}\biggr)^{m}\biggl(\frac{c+p_i}{c-p_i}\biggr)^{l},
~~A_{ij}=\frac{(p_i-p_j)^2}{(p_i+p_j)^2},
\label{PWF-AKP-sym}
\end{equation}
agrees with the symmetries \eqref{sym}.
Such a  $\tau$ function provides soliton solutions for the dAKP equation \eqref{AKP}
as well as for the dBKP equation \eqref{BKP} in light of Theorem \ref{T:1}.
It is remarkable that \eqref{PWF-AKP-sym} cannot be obtained from
the dBKP plane wave factor and phase factor \eqref{PWF-BKP} by
imposing constraints on $(p_j, q_j)$.
In the following we will go through more forms of  $\tau$ functions with the desired symmetries.

\subsection{Gramian form via Cauchy matrix approach}\label{sec-3-1}

In order to derive a more general  $\tau$ function, we construct it by means of Cauchy
matrix approach (cf.\cite{NijAH-JPA-2009,ZhaZ-SAPM-2013,HieJN-book-2016}).
Consider the  Sylvester equation
\begin{equation}\label{syl-Eq}
    \boldsymbol{KM}+\bM \bK=\br \bc^{\mathrm{T}},
\end{equation}
where $\bK$ is a given $N\times N$ constant matrix and $\bK$ and $-\bK$ do not share eigenvalues,
$\bM\in \mathbb{C}_{N\times N}[n,m,l,\cdots]$,
$\br=(r_1,r_2,\cdots, r_N)^{\mathrm{T}}$ and $r_i$ are functions $r_i: \mathbb{Z}^{\infty}\mapsto \mathbb{C}$,
and $\bc=(c_1, c_2,\cdots, c_N)^{\mathrm{T}}$ with $c_i \in \mathbb{C}$.
Dispersion relations are introduced through
\begin{subequations}\label{DR}
\begin{align}
   & \t{\boldsymbol r}=(a-\boldsymbol K)^{-1}(a+\boldsymbol K)\boldsymbol r,  \\
   & \h{\boldsymbol r}=(b-\boldsymbol K)^{-1}(b+\boldsymbol K)\boldsymbol r,    \\
   & \b{\boldsymbol r}=(c-\boldsymbol K)^{-1}(c+\boldsymbol K)\boldsymbol r.
\end{align}
\end{subequations}
Since $\bK$ and $-\bK$ do not have common eigenvalues, $\bM$ is uniquely determined by
\eqref{syl-Eq} with give $(\bK, \br, \bc)$ \cite{Sylvester}.
Here for the spacing parameters $a,b,c$ we skip the unit matrix $\bI$ without any confusion in notations.
Then, it can be proved that $\bM$ obeys the following shift evolutions \cite{ZhaZ-SAPM-2013}
\bsb
    \begin{align}
        \bs{\widetilde M}(a+\bK)-(a+\bK)\bM=\t{\br} \bc^{\mathrm{T}}, \\
        (a-\bK)\t{\bM}-\bM (a-\bK)=\br \bc^{\mathrm{T}},
    \end{align}
\esb
and the parallel relations for $(m, b)$ and $(l,c)$.
Then we introduce scalar functions
\begin{subequations}\label{S}
\begin{align}
S^{(i,j)}&=\bc^{\mathrm{T}} \bK^j(\bI+\bM)^{-1}\bK^i \br, \label{S-ij}    \\
S(\alpha,j]&=\bc^{\mathrm{T}} \bK^j (\bI+\bM)^{-1}(\alpha+\bK)^{-1}\br,    \\
S[i,\beta)&=\bc^{\mathrm{T}}(\beta+\bK)^{-1}(\bI+\bM)^{-1}\bK^i \br,
\end{align}
\end{subequations}
with $i,j\in \mathbb{Z}$,  $\alpha, \beta\in \mathbb{C}$, and define $\tau$ function as
\begin{equation}
\tau =\bs{|I+M|}.
\label{tau-G}
\end{equation}
One can derive evolutions of these functions.
$S=(S^{(i,j)})_{\infty\times \infty}$ is a symmetric matrix \cite{ZhaZ-SAPM-2013},
i.e. $S^{(i,j)}=S^{(j,i)}$.
Evolutions for  $S^{(i,j)}$ can be described as (cf.\cite{NijAH-JPA-2009,ZhaZ-SAPM-2013})
\bsb\label{S-ij-evol}
    \begin{align}
        a\widetilde{S}^{(i,j)}-\widetilde{S}^{(i,j+1)}&=aS^{(i,j)}+S^{(i+1,j)}-S^{(0,j)}\widetilde{S}^{(i,0)},   \\
        a{S}^{(i,j)}+{S}^{(i,j+1)}&=a\widetilde S^{(i,j)}+\widetilde S^{(i+1,j)}-\widetilde S^{(0,j)}{S}^{(i,0)},
\end{align}
\esb
and parallel relations for $(m,b)$ and $(l,c)$.
One more symmetric relation is $S(\alpha,j]=S[j,\alpha)$ due to $S^{(i,j)}=S^{(j,i)}$ and  the expansion
\[        S(\alpha,j]=\sum_{i=0}^{\infty}\frac{(-1)^i}{a^{i+1}}S^{(i,j)},  ~~
        S[j,\alpha)=\sum_{i=0}^{\infty}\frac{(-1)^i}{a^{i+1}}S^{(j,i)}.
\]
With the above expansion and \eqref{S-ij-evol} one can obtain  evolutions of $S(\alpha,j]$ as
    \bsb\label{saj-move}
    \begin{align}
        \label{saj-move-1}a\widetilde{S}(\alpha,j]-\widetilde{S}(\alpha,j+1]
        &=(a-\alpha)S(\alpha,j]+S^{(0,j)}(1-\widetilde S(\alpha,0]), \\
        \label{saj-move-2}b\widehat{S}(\alpha,j]-\widehat{S}(\alpha,j+1]
        &=(b-\alpha)S(\alpha,j]+S^{(0,j)}(1-\widehat S(\alpha,0]), \\
        \label{saj-move-3}c\overline{S}(\alpha,j]-\overline{S}(\alpha,j+1]
        &=(c-\alpha)S(\alpha,j]+S^{(0,j)}(1-\overline S(\alpha,0]).
    \end{align}
    \esb

Next, let us derive some relations on $\tau$ and the functions in \eqref{S} (cf.\cite{HieJN-book-2016}).

\begin{lemma}\label{L:1}
    $\tau$ function \eqref{tau-G} obeys  evolutions
    \bsb
    \begin{align}
        \label{24a}\t\tau/\tau&=1-S[0,-a),   \\
        \label{24b}\tau/\t\tau&=1-\widetilde S[0,a).
    \end{align}
    \esb
    and the parallel relations for $(m,b)$ and $(l,c)$.
\end{lemma}

\begin{proof}
    First,
\[|a-\bK|\t\tau =|a-\bK||\bI+\widetilde\bM| =|a-\bK+\bM(a-\bK)+\br \bc^{\mathrm{T}} |,\]
i.e.,
\[ \t\tau =|\bI+\bM+\br\bc^{\mathrm{T}}(a-\bK)^{-1}|
 =|\bI+\bM||\bI+\br\bc^{\mathrm{T}}(a-\bK)^{-1}(\bI+\bM)^{-1}|.   \]
Then, by means of Weinstein-Aronszajn identity $|\bI+\bs r\bs c^{\mathrm{T}}|=1+\bs c^{\mathrm{T}}\bs r$,
one has
\begin{align*}
 \t\tau/\tau =1-\bc^{\mathrm{T}}(\bK-a)^{-1}(\bI+\bM)^{-1}\br    =1-S[0,-a).
\end{align*}

\eqref{24b} can be proved as the following.
\[|a+\bK|\tau =|a+\bK||\bI+\bM| =|a+\bK+\widetilde\bM(a+\bK)-\t\br \bc^{\mathrm{T}} |, \]
and
\[\tau =|\bI+\widetilde\bM-\t\br\bc^{\mathrm{T}}(a+\bK)^{-1}|
=|\bI+\widetilde\bM||\bI-\t\br\bc^{\mathrm{T}}(a+\bK)^{-1}(\bI+\widetilde\bM)^{-1}|.  \]
It then follows from  the  Weinstein-Aronszajn identity that one gets \eqref{24b}.
\end{proof}

\begin{lemma}\label{L:2}
    Setting $u=S^{(0,0)}$, we  have the following relations
    \bsb\label{u-t}
    \begin{align}
        \label{u-t1}(a-b+\h u-\t u)   &=(a-b)\frac{\th \tau \tau}{\t \tau\h \tau},  \\
        \label{u-t2}(b-c+\b u-\h  u)     &=(b-c)\frac{\h{\b \tau} \tau}{\h \tau\b\tau},        \\
        \label{u-t3}(c-a+\t  u-\b  u)   &=(c-a)\frac{\t{\b \tau} \tau}{\b \tau\t \tau}.
    \end{align}
    \esb
\end{lemma}

\begin{proof}
The proof has been give in \cite{NijAH-JPA-2009} for the case $\bK$ is diagonal.
Here let us extract out main steps and extend them to the case of arbitrary $\bK$.

Considering the evolutions of $S(\alpha,j]$ given in \eqref{saj-move-1} and \eqref{saj-move-2}
where we take  $\alpha=a,j=0$, we get
    \bsb\label{26}
    \begin{align}
        \label{26a}a\widetilde S(a,0]-\widetilde S(a,1]&=u(1-\widetilde S(a,0]), \\
        \label{26b}b\widehat S(a,0]-\widehat S(a,1]&=(b-a)S(a,0]+u(1-\th S(a,0]).
    \end{align}
    \esb
Eliminating $S(a,1]$ from the above gives rise to
    \begin{align}
        (a-b+\h u-\t u)(1-\th S(a,0])-(a-b)(1-\t S(a,0])=0.
    \end{align}
which indicates the relation  \eqref{u-t1} using \eqref{24b}.
\eqref{u-t2} and \eqref{u-t3} can be obtained in a similar way.
\end{proof}

Thus, combining the three equations in \eqref{u-t} together, we get the dAKP equation.

\begin{theorem}\label{T:2}
The  $\tau$ function defined in \eqref{tau-G} satisfies the dAKP equation \eqref{AKP}.
\end{theorem}

Solution $\bM$ to the Sylvester equation \eqref{syl-Eq} and $\br$ to the dispersion relation \eqref{DR}
can be written out in terms of the canonical forms of $\bK$.
For a given $\bK$, the dispersion relation \eqref{DR} indicates the symmetry for $\br$:  $\br((n_i,a_i))=\br((-n_i,-a_i))$,
so is for $\bM$, i.e.  $\bM((n_i,a_i))=\bM((-n_i,-a_i))$,
and so is for $\tau$, i.e.  $\tau((n_i,a_i))=\tau((-n_i,-a_i))$.

Canonical form of $\bK$ is composed of a diagonal matrix and different Jordan blocks.
When $\bK$ is a diagonal matrix $\bK=\mathrm{diag}\{p_1, p_2, \cdots, p_N\}$,
$\br$ consists of $r_i=\eta_i$ where $\eta_i$ is given in \eqref{PWF-AKP-sym},
and $\bM=(m_{ij})_{N\times N}$ consists of
\begin{equation}
m_{ij}=\frac{r_ic_j}{p_i+p_j}.
\end{equation}
Then, $\tau=|\bI+\bM|$ is a Gramian which is a special case of the solution obtained in \cite{OhtHTI-JPSJ-1993}.
When $\bK$ is a Jordan block and a more general form, one can refer to  Sec.4 of  \cite{ZhaZ-SAPM-2013}
for explicit forms of $\br$ and $\bM$.

\subsection{Casoratian solutions}\label{sec-3-2}

The deformed dAKP equation \eqref{ssym-b} also appeared as a member in the bilinear forms of H3 equation
in the Adler-Bobenko-Suris (ABS) list \cite{AdlBS-CMP-2003}, (see Eq.(5.20a) in \cite{HieZ-JPA-2009}).
It allows a Casoratian solution \cite{HieZ-JPA-2009}
\begin{subequations}\label{tau-C+psi}
\begin{equation}
\tau(\psi)=|\psi(n,m,l), \psi(n,m,l+1), \psi(n,m,l+2), \cdots, \psi(n,m,l+N-1)|,
\label{tau-C}
\end{equation}
where $\psi=(\psi_1, \psi_2, \cdots, \psi_N)^{\mathrm{T}}$ and
\begin{equation}
\psi_i=\gamma^{+}_i(a+p_i)^n(b+p_i)^m(c+p_i)^l+\gamma^{-}_i(a-p_i)^n(b-p_i)^m(c-p_i)^l,
~ \gamma^{\pm}_i\in \mathbb{C}.
\label{psi}
\end{equation}
\end{subequations}
At the first glance, $\tau(\psi)$ does not have symmetries \eqref{sym}.
However, by means  of the gauge property of Hirota's discrete bilinear equations,
$\tau(\psi)$ does satisfy the dAKP and dBKP equations simultaneously.

\begin{theorem}\label{T:3}
The $\tau$ function $\tau(\psi)$ defined in \eqref{tau-C+psi} is a solution of the dAKP equation \eqref{AKP}
as well as the dBKP equation \eqref{BKP}.
\end{theorem}

\begin{proof}
In addition to $\psi$, we introduce $N$-th order column vectors $\varphi, \omega$ and $\theta$ that are composed of,
respectively, (cf.\cite{HieZ-JPA-2009})
\begin{subequations}
\begin{align}
& \varphi_i=\gamma^{+}_i(a-p_i)^{-n}(b+p_i)^m(c+p_i)^l+\gamma^{-}_i(a+p_i)^{-n}(b-p_i)^m(c-p_i)^l,\label{phi}\\
& \omega_i=\gamma^{+}_i(a+p_i)^n(b-p_i)^{-m}(c+p_i)^l+\gamma^{-}_i(a-p_i)^n(b+p_i)^{-m}(c-p_i)^l,\label{omega}\\
& \theta_i=\gamma^{+}_i(a+p_i)^n(b+p_i)^m(c-p_i)^{-l}+\gamma^{-}_i(a-p_i)^n(b-p_i)^m(c+p_i)^{-l}.\label{theta}
\end{align}
\end{subequations}
One can prove that (cf.\cite{HieZ-JPA-2009})
\begin{align*}
\tau(\psi((n,a),(m,b),l))& =A^n\,\tau(\varphi((n,a),m,l))= (-1)^{N\times n} A^n\,  \tau(\psi((-n,-a),m,l))\\
                                     &  =B^m\,\tau(\omega(n,(m,b),l)= (-1)^{N\times m} B^m\,\tau(\psi(n,(-m,-b),l)),
\end{align*}
where $A=\prod^{N}_{i=1}(a^2-p_i^2)$ and $B=\prod^{N}_{i=1}(b^2-p_i^2)$.
This means that $\tau(\psi)$ satisfies the extended symmetries \eqref{sym'} in $n$ and $m$-direction.
Meanwhile, note that due to the relation $\t\psi-\b\psi=(a-c)\psi$,
the $\tau$ function \eqref{tau-C} can be equivalently constructed in terms of shifts of $n$, i.e.
\begin{equation}
\tau(\psi)=|\psi(n,m,l), \psi(n+1,m,l), \psi(n+2,m,l), \cdots, \psi(n+N-1,m,l)|,
\label{tau-C-n}
\end{equation}
(see Eq.(2.24) in \cite{HieZ-JPA-2009}).
With this notation,
\[\tau(\psi(n,m,(l,c)))   = C^l\,\tau(\theta(n,m,(l,c))=(-1)^{N\times l} C^l\, \tau(\psi(n,m,(-l,-c))),\]
where $C=\prod^{N}_{i=1}(c^2-p_i^2)$.
This gives    the extended symmetries \eqref{sym'-c}.
Thus, due to the gauge property of Hirota's discrete bilinear equations,
\eqref{tau-C+psi} allows symmetries \eqref{sym'}
and consequently provides  a solution for both the dAKP and dBKP equations.
\end{proof}

Multiple pole solutions of the deformed dAKP equation \eqref{ssym-b} is given by $\tau(\psi)$ in the Casoratian
form \eqref{tau-C}, but where
\begin{equation}
\psi_1=(\ref{psi})|_{i=1},~~
\psi_j=\frac{1}{(j-1)!}\partial^{j-1}_{p_1}\psi_1, ~(j=2, 3, \cdots).
\label{psi'}
\end{equation}
This can be found in Theorem 1 in \cite{ShiZ-SIGMA-2011}.
We claim that $\tau(\psi)$ with \eqref{psi'} satisfies the extended symmetries \eqref{sym'} and then
it solves the dAKP and the dBKP as well. To elaborate this, we introduce
lower triangular Toeplitz matrix (LTTM)
\begin{equation}
\bT=\left(\begin{array}{ccccc}
t_0 & 0 & \cdots & 0 & 0\\
t_1 & t_0 & \cdots & 0 & 0\\
\vdots & \vdots & \ddots & \vdots & \vdots\\
t_{N-2} & t_{N-3} & \cdots & t_0 & 0\\
t_{N-1} & t_{N-2} & \cdots  & t_1 &  t_0
\end{array}
\right)
\end{equation}
and note that such a matrix can be generated by some function $f(p)$ via taking
\begin{equation}
t_j=\frac{1}{j!}\partial_p^{j} f(p), ~~ j=0, 1,\cdots.
\label{tj}
\end{equation}
For convenience, by $\bT[f(p)]$ we denote a LTTM generated by the function $f(p)$ via  \eqref{tj}.
We also introduce new auxiliary vectors $\varphi, \omega$ and $\theta$ by
\begin{subequations}
\begin{align}
& \varphi_1=(\ref{phi})|_{i=1},~~ \varphi_j=\frac{1}{(j-1)!}\partial^{j-1}_{p_1}\varphi_1,  \label{phi'}\\
& \omega_1=(\ref{omega})|_{i=1},~~ \varphi_j=\frac{1}{(j-1)!}\partial^{j-1}_{p_1}\omega_1,  \label{omega'}\\
& \theta_1=(\ref{phi})|_{i=1},~~ \theta_j=\frac{1}{(j-1)!}\partial^{j-1}_{p_1}\theta_1,  \label{theta'}
\end{align}
\end{subequations}
for $j=2,3,\cdots$.
Note that these vectors are connected to $\psi$ composed of \eqref{psi'} by
\begin{equation}
\psi=\bT[a^2-p_1^2]\varphi =\bT[b^2-p_1^2]\omega =\bT[c^2-p_1^2]\theta.
\end{equation}
Then we have
\begin{align*}
\tau(\psi((n,a),(m,b),(l,c)))& =\bT[a^2-p_1^2]\,\tau(\varphi((n,a),m,l))=(p_1^2-a^2)^n \tau(\psi((-n,-a),m,l))\\
                                     &  =\bT[b^2-p_1^2]\, \tau(\omega(n,(m,b),l)=(p_1^2-b^2)^m \tau(\psi(n,(-m,-b),l))\\
                                     &  =\bT[c^2-p_1^2]\, \tau(\theta(n,m,(l,c))=(p_1^2-c^2)^l \tau(\psi(n,m,(-l,-c))),
\end{align*}
which are in the form of the extended symmetries \eqref{sym'}.

Let us sum up the above discussion by the following theorem.
\begin{theorem}\label{T:4}
The function $\tau(\psi)$ composed of \eqref{psi'} provides a multiple pole solution
to the dAKP equation \eqref{AKP} as well as the dBKP equation \eqref{BKP}.
\end{theorem}

\subsection{Polynomial solutions}\label{sec-3-3}

The dAKP equation has polynomial solutions (cf.\cite{Nimmo-JPA-1997}).
Explicit form of these solutions can be described as the following.
\begin{lemma}\label{L:3}\cite{ZhaZ-JNMP-2019}
Let
\begin{equation}
\psi^+_0 =\varrho_0 (1+p/a)^n(1+p/b)^m(1+p/c)^l(1+p)^s,
\label{psi-0}
\end{equation}
where
\[     \varrho_0=  \frac{1}{2}\mathrm{e}^{-\sum^{\infty}_{j=1}\frac{(- p )^{j}}{j}\gamma_j},
~~ \gamma_j\in \mathbb{C}.\]
Then,
\begin{align*}
\psi_0^+= &  \frac{1}{2}\sum^{\infty}_{j=0}\phi_j p^j,~~
(\phi_j= \frac{2}{j!}\partial^{j}_{p} \psi^{+}_0|_{p=0})\\
=& \frac{1}{2}\exp \bigg[{-\sum^{\infty}_{j=1}\frac{(-p)^{j}}{j} \check{x}_j }  \bigg],
\end{align*}
where
\begin{equation}
\check{x}_j=x_j+s,~~ x_j=n a^{-j} +m b^{-j}+l c^{-j} + \gamma_j.
\label{x-j}
\end{equation}
$\phi_j=\phi_j(n,m,l,s)$ can be expressed in terms of $x_j$ by
\begin{equation}
\phi_j = (- 1)^j\sum_{||\mu||=j}(-1)^{|\mu|}\frac{\check{\mathbf{x}}^{\mu}}{\mu!}
\end{equation}
where
\begin{align*}
&\mu =(\mu_1,\mu_2,\cdots),~~ \mu_j\in \{0, 1, 2,\cdots\},~~ ||\mu||=\!\sum^{\infty}_{j=1}j\mu_j,  \\
&|\mu|=\!\sum^{\infty}_{j=1}\mu_j,~~\mu ! =\mu_1!\cdot \mu_2!\cdots,~~
{\check{\mathbf{x}}}^{\mu}\! =\left(\frac{\check x_1}{1}\right)^{\mu_1}\left( \frac{\check x_2}{2}\right)^{\mu_2}\cdots.
\end{align*}
The first few  $\phi_j$ are
\begin{align*}
& \phi_0=1,~~ \phi_1=\check{x}_1,~~ \phi_2= \frac{1}{2}(\check{x}_1^2 -   \check{x}_2),~~
  \phi_3=\frac{1}{6}(\check{x}_1^3 - 3 \check{x}_1 \check{x}_2 + 2 \check{x}_3),\\
& \phi_{4}=\frac{1}{24}(\check{x}_1^4-6\check{x}_1^2\check{x}_2
+ 8 \check{x}_1\check{x}_3 +3 \check{x}_2^2 -6 \check{x}_4).
\end{align*}
Define
\begin{equation}
\phi=(\phi_1, \phi_3, \phi_5, \cdots, \phi_{2N-1})^{\mathrm{T}}.
\label{pphi}
\end{equation}
The Casoratian
\begin{equation}
\tau_N(\phi)=|\phi(n,m,l,0), \phi(n,m,l,1), \phi(n,m,l,2), \cdots, \phi(n,m,l,N-1)|
\label{tau-p}
\end{equation}
is a solution of the deformed dAKP equation \eqref{ssym-c} (i.e. Eq.(3.15) in \cite{ZhaZ-JNMP-2019}).
\end{lemma}

Note that $\tau_N(\phi)$ satisfies the superposition formula \cite{ZhaZ-SIGMA-2017}
\begin{equation}
\tau_{N-1}(E_{n_i} \tau_{N+1}) -\tau_{N+1}(E_{n_i} \tau_{N-1})=\frac{1}{a_i}\tau_{N}(E_{n_i} \tau_{N}),
\end{equation}
and provides a discrete analogue of the remarkable Burchnall-Chaundy polynomials (cf.\cite{VesW-JPA-2015}).

Let us look at symmetries of $\tau_N(\phi)$
The first three are
\[
\tau_{1}(\phi)=x_1,  ~~ \tau_{2}(\phi)=\frac{x_1^3-x_3}{3}, ~~
\tau_{3}(\phi)=\frac{1}{45}x_1^6-\frac{1}{9}x_1^3x_3+\frac{1}{5}x_1x_5-\frac{1}{9}x_3^2, \]
which only depend on $\{x_{2i+1}\}$.  For general $N$, it has been proved that
(in Appendix C of \cite{ZhaZ-SIGMA-2017})
\begin{lemma}\label{L:4}
 $\tau_N(\phi)$ depends only on $\{x_1, x_3, \cdots, x_{2N-1}\}$.
\end{lemma}

Thus, from the definition \eqref{x-j} of $x_j$, we immediately find that  $\tau_N(\phi)$ has symmetries \eqref{sym}.
This then leads to polynomial solutions of the dBKP equation.

\begin{theorem}\label{T:5}
The $\tau$ function $\tau(\phi)$ defined by \eqref{tau-p} provides polynomial solutions
for the dAKP equation \eqref{AKP} as well as the dBKP equation \eqref{BKP}.
\end{theorem}

\section{Elliptic case}\label{sec-4}

\subsection{dAKP}

Eq.(2.51) in \cite{YooN-JMP-2013} is a version of  dAKP equation ready for elliptic solitons.
It is written as
\begin{equation}\label{AKP-e}
     \mathcal{E}_0\doteq  \Phi_{a}^{-b}\,\th \tau \b \tau+ \Phi_{b}^{-c}\hb \tau \t \tau  + \Phi_{c}^{-a} \tb \tau \h \tau=0,
\end{equation}
where
\begin{equation}\label{Phi-sig}
\Phi_a(b)=\Phi_a^b=\frac{\sigma(a+b)}{\sigma(a) \sigma(b)}.
\end{equation}
Here and below, $\sigma(z)$, $\zeta(z)$ and $\wp(z)$ are the Weierstrass functions.
Eq.\eqref{AKP-e} can also be written as
\begin{equation}\label{AKP-e-alt}
        \sigma(c) \sigma(a-b)\th\tau \b\tau+ \sigma(a) \sigma(b-c)\hb \tau \t \tau+ \sigma(b) \sigma(c-a)\tb \tau \h \tau=0.
\end{equation}
Note that this is similar to \eqref{AKP-alt}, not to \eqref{AKP}.

The following $\tau$ function is given  as an elliptic soliton solution of the dAKP equation
\eqref{AKP-e},\cite{YooN-JMP-2013}
\begin{equation}\label{tau-ell}
\tau=\sigma(\xi)|\bI+\mathcal{M}|
\end{equation}
where
\begin{equation}
\xi=an+bm+cl+\xi_0  , ~ \xi_0\in \mathbb{C},
\label{xi}
\end{equation}
\begin{subequations}
\begin{align}
&\mathcal{M}=(\mathcal{M}_{ij})_{N\times N},~~ \mathcal{M}_{ij}=\rho_i \mathcal{M}^0_{ij}\nu_j,\\
& \rho_i=(\Phi_a(-\kappa_i))^n (\Phi_b(-\kappa_i))^m(\Phi_c(-\kappa_i))^l
\mathrm{e}^{\zeta(\xi) \kappa_i }\rho^0 (\kappa_i), \\
&\nu_j=(\Phi_a(\kappa'_j))^{-n} (\Phi_b(\kappa'_j))^{-m}(\Phi_c(\kappa'_j))^{-l}
\mathrm{e}^{\zeta(\xi) \kappa'_j }\nu^0 (\kappa'_j), \\
&\mathcal{M}^0_{ij}= \Phi_\xi(\kappa_i+\kappa'_j) \mathrm{e}^{-\zeta(\xi) (\kappa_i+\kappa'_j) },
\end{align}
\end{subequations}
and $\kappa_i,\kappa'_j\in \mathbb{C}$ for $i, j=1,2,\cdots, N$.

\subsection{Symmetries of the $\tau$ function}

To get a $\tau$ function that allows symmetries \eqref{sym}, we take
\begin{equation}\label{cond-sym}
\kappa'_j=\kappa_j, ~~ \rho^0 (\kappa_j)=\nu^0 (\kappa_j),~~ (j=1,2,\cdots, N).
\end{equation}
Obviously, $\xi(-n,-a)=\xi(n,a)$. In addition, with \eqref{cond-sym} we have
\begin{subequations}\label{Mij-sym}
\begin{equation}
\mathcal{M}_{ij}=S_i\Phi_\xi(\kappa_i+\kappa_j) T_j,
\end{equation}
where
\begin{align}
& S_i=(\Phi_a(-\kappa_i))^n (\Phi_b(-\kappa_i))^m(\Phi_c(-\kappa_i))^l \rho^0 (\kappa_i), ~(i=1,2,\cdots, N), \\
&T_j=(\Phi_a(\kappa_j))^{-n} (\Phi_b(\kappa_j))^{-m}(\Phi_c(\kappa_j))^{-l}\rho^0 (\kappa_j), ~(j=1,2,\cdots, N).
\end{align}
\end{subequations}
Noticing that
\[S_i(-n,-a)=A_i^n S_i(n,a),~~ T_i(-n,-a)=A_i^{-n} T_i(n,a),\]
where
\[A_i=-\frac{\sigma^2(a)\sigma^2(\kappa_i)}{\sigma(\kappa_i+a)\sigma(\kappa_i-a)} = \frac{1}{\wp(\kappa_i)-\wp(a)},\]
we then have
\begin{align*}
\tau(-n,-a)&=\sigma(\xi) |\bI+\mathcal{M}|_{(-n,-a)}\\
 &=\sigma(\xi) \mathrm{Det}[\mathrm{Diag}(A_1^n, A_2^n, \cdots, A_N^n)(\bI+\mathcal{M})_{(n,a)}
 \mathrm{Diag}(A_1^{-n}, A_2^{-n}, \cdots, A_N^{-n})]\\
 &= \sigma(\xi)|\bI+\mathcal{M}|_{(n,a)}=\tau(n,a).
\end{align*}
Since the symmetries w.r.t. $m$ and $l$ can be proved similarly, one can conclude that
\begin{lemma}\label{L:5}
 The function $\tau=\sigma(\xi)|\bI+\mathcal{M}|$, where $\mathcal{M}_{ij}$ are defined in \eqref{Mij-sym},
 allows symmetries \eqref{sym}.
\end{lemma}

\subsection{dBKP}

Now that the dAKP $\tau$ function defined in Lemma \ref{L:5} allows symmetries \eqref{sym},
as the counterparts of \eqref{ssym}, from \eqref{AKP-e} we have deformations
\begin{subequations}\label{ssym-e}
\begin{align}
& \mathcal{E}_1\doteq \Phi_{b}^{-c} \,\th{\b \tau}\tau+ \Phi_{c}^{a}\th \tau \b\tau -\Phi_a^b\,\tb \tau \h\tau=0,
\label{ssym-ea}\\
& \mathcal{E}_2\doteq \Phi_c^{-a}\,\th{\b \tau}\tau+\Phi_a^b \hb \tau \t \tau-\Phi_{b}^{c} \th \tau\b \tau=0,
\label{ssym-eb}   \\
& \mathcal{E}_3\doteq \Phi_a^{-b}\,\th{\b \tau}\tau+\Phi_{b}^{c} \tb \tau \h \tau-\Phi_{c}^{a} \hb \tau \t \tau=0,
\label{ssym-ec}
\end{align}
\end{subequations}
where we have made use of $\Phi_a^b=\Phi_b^a=-\Phi_{-a}^{-b}$.
To derive a dBKP equation with elliptic coefficients,
multiplying $\Phi_c^{-a}\Phi_a^{-b}$, $\Phi_a^{-b}\Phi_{b}^{-c}$,
$\Phi_{b}^{-c}\Phi_c^{-a}$ to the three equations in \eqref{ssym-e}, respectively,
and summing them together, we get
\begin{align}
&  \Phi_c^{-a}\Phi_a^{-b}\times \mathcal{E}_1 +\Phi_a^{-b}\Phi_{b}^{-c}\times \mathcal{E}_2
+\Phi_{b}^{-c}\Phi_c^{-a}\times \mathcal{E}_3 \nonumber\\
=~& 3 \Phi_a^{-b}\Phi_{b}^{-c} \Phi_c^{-a} \th{\b \tau} \tau
+ \Phi_a^{-b}\left(\Phi_{c}^{a} \Phi_c^{-a}-\Phi_{b}^{c} \Phi_b^{-c}\right) \th \tau \b \tau  \nonumber  \\
& + \Phi_b^{-c}\left(\Phi_{a}^{b} \Phi_a^{-b}-\Phi_{c}^{a} \Phi_c^{-a}\right)\hb\tau \t \tau
 + \Phi_c^{-a} \left(\Phi_{b}^{c} \Phi_b^{-c}-\Phi_{a}^{b} \Phi_a^{-b} \right)\tb\tau \h \tau =0,
\label{BKP-e1}
\end{align}
i.e.
\begin{align}
& 3 \Phi_a^{-b}\Phi_{b}^{-c} \Phi_c^{-a} \th{\b \tau} \tau
+ \Phi_a^{-b} (2\wp(c)-\wp(a)-\wp(b) ) \th \tau \b \tau  \nonumber  \\
& + \Phi_b^{-c}  (2\wp(a)-\wp(b)-\wp(c) ) \hb\tau \t \tau
 + \Phi_c^{-a} (2\wp(b)-\wp(a)-\wp(c) )  \tb\tau \h \tau = 0,
\label{BKP-e2}
\end{align}
where we have made use of
\begin{equation} \label{wp}
\Phi_x^y \Phi_x^{-y}=\wp(x)-\wp(y).\end{equation}
Finally, multiplying $ (\wp(a)+\wp(b)+\wp(c))$ to the dAKP \eqref{AKP-e} and then adding it to the above equation, we get
\begin{equation}
\mathcal{F}\doteq \Phi_a^{-b}\Phi_{b}^{-c} \Phi_c^{-a} \th{\b \tau} \tau
+ \Phi_a^{-b}  \wp(c)  \th \tau \b \tau  + \Phi_b^{-c} \wp(a)  \hb\tau \t \tau
 + \Phi_c^{-a} \wp(b)   \tb\tau \h \tau = 0,
\label{BKP-ee}
\end{equation}
which is a dBKP equation with elliptic coefficients.
Noting that
\begin{equation}\label{BKPF}
3\mathcal{F}= \Phi_c^{-a}\Phi_a^{-b}\times \mathcal{E}_1 +\Phi_a^{-b}\Phi_{b}^{-c}\times \mathcal{E}_2
+\Phi_{b}^{-c}\Phi_c^{-a}\times \mathcal{E}_3
+(\wp(a)+\wp(b)+\wp(c))\times \mathcal{E}_0,
\end{equation}
we immediately have the following.

\begin{theorem}\label{T:6}
Both the dAKP equation \eqref{AKP-e} and dBKP equation \eqref{BKP-ee} with elliptic coefficients
allow a solution $\tau$ defined in Lemma \ref{L:5}.
\end{theorem}

By means of the transformation like \eqref{gauge-bkp}, we can transform \eqref{BKP-ee} to be of
arbitrary nonzero coefficients. We have the following.
\begin{theorem}\label{T:7}
Substituting
\begin{subequations}\label{gauge-bkp-ee}
\begin{equation}
\tau = A^{\frac{ml}{2}}B^{\frac{nl}{2}}C^{\frac{nm}{2}}\times f
\end{equation}
into \eqref{BKP-ee},
where
\begin{equation}\label{ABC}
A=\frac{\wp(b)\wp(c)}{\wp(a)(\Phi_b^c)^2}, ~~
B=\frac{\wp(c)\wp(a)}{\wp(b)(\Phi_c^a)^2},~~
C=\frac{\wp(a)\wp(b)}{\wp(c)(\Phi_a^b)^2},
\end{equation}
\end{subequations}
one can transform Eq.\eqref{BKP-ee} to the following dBKP equation,
\begin{equation}
 \Phi_a^{-b}\Phi_{b}^{-c} \Phi_c^{-a} \th{\b f} f
+ \Phi_a^{-b} \Phi_{b}^{c} \Phi_c^{a}   \th f \b f  + \Phi_b^{-c} \Phi_a^{b} \Phi_c^{a}  \hb f \t f
 + \Phi_c^{-a}  \Phi_a^{b}\Phi_{b}^{c}  \tb f \h f = 0.
\label{BKP-eee}
\end{equation}
Elliptic solutions of this dBKP equation are given by
\begin{equation}\label{tau-ell-eee}
f=A^{-\frac{ml}{2}}B^{-\frac{nl}{2}}C^{-\frac{nm}{2}}\sigma(\xi)|\bI+ \mathcal{M} |,
 \end{equation}
where $\mathcal{M}_{ij}$ are defined in \eqref{Mij-sym} and $A,B,C$ are given in \eqref{ABC}.
\end{theorem}

\begin{remark}\label{R:3}
$\mathcal{E}_i$ $(i=0,1,2,3)$ can be expressed as linear combinations of any two elements of
$\{\mathcal{E}_0, \mathcal{E}_1, \mathcal{E}_2, \mathcal{E}_3\}$,
for example,
\begin{align}
 \Phi_{a}^{b}\mathcal{E}_0=-\Phi^{-a}_{c}\mathcal{E}_1+ \Phi^{-c}_{b}\mathcal{E}_2 ,
 ~~ \Phi_{a}^{b}\mathcal{E}_3=-\Phi_{b}^{c}\mathcal{E}_1- \Phi_{c}^{a}\mathcal{E}_2.\label{ck1}
\end{align}
This also means there are alternative expressions of \eqref{BKPF}, e.g.
\begin{align*}
& \Phi_{a}^{b} \mathcal{F}  =-\Phi^{-a}_{c}\wp(b)\mathcal{E}_1+ \Phi^{-c}_{b}\wp(a)\mathcal{E}_2,\\
&  \mathcal{F}  =\wp(a)\mathcal{E}_0+ \Phi^{-a}_{c}\Phi^{-b}_{a}\mathcal{E}_1.
\end{align*}
Note that to obtain \eqref{ck1} we need to make use of a special case $(v=0)$ of the well-known identity
\begin{align}
&\sigma(x+y) \sigma(x-y) \sigma(u+v) \sigma(u-v) \nonumber \\
=~& \sigma(x+u) \sigma(x-u)
\sigma(y+v) \sigma(y-v)
- \sigma(x+v) \sigma(x-v) \sigma(y+u) \sigma(y-u).
\label{sigma}
\end{align}
\end{remark}

Finally, let us back to the equation dAKP \eqref{AKP-e-alt} and dBKP \eqref{BKP-ee}.
Both of them have the simplest solution
$\tau=\sigma(\xi)$ where $\xi$ is given in \eqref{xi}.
In this case \eqref{AKP-e-alt} is related to the  identity \eqref{sigma}
by (cf.\cite{NijA-IMRN-2010,YooN-JMP-2013})
\[x=\frac{1}{2}(a-b+c),~y=\frac{1}{2}(c-a+b), ~u=\xi+\frac{1}{2}(a+b+c),~ v=\frac{1}{2}(a+b-c),\]
or, equivalently,
\begin{equation}
a=x+v,~ b=y+v,~ c=x+y, ~ \xi=u-x-y-v.
\label{abcxi}
\end{equation}
When $\tau=\sigma(\xi)$,
since the three equations in \eqref{ssym-e} that are used to derive \eqref{BKP-ee}
are essentially the identity \eqref{sigma} with reparameters of $x,y,u$ and $v$,
we can substitute \eqref{abcxi} into the dBKP \eqref{BKP-ee} where $\tau=\sigma(\xi)$,
and we arrive at the the following equality
\begin{align}
& \sigma(x-y)\sigma(v-x)\sigma(y-v)\sigma(u+x+y+v)\sigma(u-x-y-v) \nonumber\\
+~& \wp(x+v)\sigma^2(v+x)\sigma(y+v)\sigma(x+y)\sigma(v-x)\sigma(u-y)\sigma(u+y) \nonumber \\
+~& \wp(y+v)\sigma^2(y+v)\sigma(v+x)\sigma(x+y)\sigma(y-v)\sigma(u-x)\sigma(u+x) \nonumber \\
+~& \wp(x+y)\sigma^2(x+y)\sigma(v+x)\sigma(y+v)\sigma(x-y)\sigma(u-v)\sigma(u+v) \nonumber \\
=~& 0. \label{sigma-2}
\end{align}

The latter elliptic identity is a consequence of the derivation of the dBKP equation, relying on the
fact that $\tau=\sigma(\xi)$ is a solution, and obtained by substituting \eqref{abcxi} and
\begin{align*}
& a-b=x-y,~ b-c=v-x,~ c-a=y-v,\\
& a+b=v+(x+y+v),~ a+c=x+(x+y+v),~ b+c=y+(x+y+v),\\
& \tau=\sigma(\xi)=\sigma(u-(x+y+v)),~ \th{\b\tau}=\sigma(u+(x+y+v)),\\
& \t\tau=\sigma(u-y),~ \h\tau=\sigma(u-x),~ \b\tau=\sigma(u-v),\\
& \hb \tau=\sigma(u+y),~ \tb\tau=\sigma(u+x),~ \th\tau=\sigma(u+v).
\end{align*}
The identity \eqref{sigma-2} demonstrates that at a basic level, the discrete equations can be interpreted
as addition formulae, albeit of a special type, for the relevant functions; in the present case for the
Weierstrass $\sigma$ and $\wp$ functions.

\section{Concluding remarks}\label{sec-5}

We have shown that under  special (i.e. reflection) symmetries \eqref{sym} the dBKP equation \eqref{BKP}
can be expressed as a linear combination of the dAKP equation \eqref{AKP}
and their reflected symmetric forms \eqref{ssym}.
This leads to a  common subset of solution spaces of the dAKP equation \eqref{AKP}
and  the dBKP equation \eqref{BKP},
which is different from the Pfaffian type link \eqref{Pfaffian} by coordinates reduction.
As argued earlier, in a sense such solutions are reducible, as they obey simultaneously two
different partial difference equations, each of which allow for in principle different solution classes.
Nonetheless, we conjecture that these solutions give an insight into the elliptic parametrisation of the
dBKP equation which was unkown hitherto.

We also checked the  case of 4D equations, comprising the higher-order equations in the relevant
dKP hierarchies. Since the 4D dAKP \eqref{4dAKP} and 4D dBKP \eqref{BKP-4D}
are direct results of the 4D consistency of the dAKP and dBKP, respectively,
the two 4D lattice equations allow symmetric $\tau$ function solutions as well.
In addition, the 4D dBKP is connected to the 4D dAKP as a linear combination of the latter and its symmetric deformations.

It is also remarkable that the plane wave factors entering in the symmetric $\tau$ functions in Sec.\ref{sec-3}
(e.g. \eqref{psi})
coincide with the plane wave factors, and their corresponding multidimensional extensions,
 of the ABS list of lattice equations, \cite{AdlBS-CMP-2003}, in the
parametrisation that allows their uniform treatment of soliton solutions, cf.\cite{NijAH-JPA-2009,ZhaZ-SAPM-2013}.
This explains why the dAKP and its reflected symmetric forms frequently
(sometimes simultaneously) appear in the bilinearisations of the ABS
equations, e.g. (4.7) in \cite{AtkHN-JPA-2008}, (5.20a,b) in \cite{HieZ-JPA-2009}, and (3.15) in \cite{ZhaZ-JNMP-2019}.
This also implies a possible yet uncovered link between the dBKP equation and ABS lattice equations.
Finally, in Sec.\ref{sec-4} we explored the elliptic version of the dBKP equation
and we obtained a  parametrisation of the  dBKP equation \eqref{BKP-ee} with elliptic coefficients
and its gauge equivalent form \eqref{BKP-eee}.  Especially in this elliptic case it would be interesting to
establish whether there are non-symmetric solutions of the KP equation that would obey a relation of the type
\eqref{Pfaffian}. This will be a subject for future investigations.

\vskip 25pt
\subsection*{Acknowledgments}
This project is  supported by the NSF of China (Nos.11875040 and 11631007) and Shanghai Sailing Program (No. 20YF1433000).

\end{document}